\documentclass[11pt,reqno]{amsart}
\usepackage{amsfonts, amsmath, amsthm, amssymb, url}
\usepackage{hyperref}
\usepackage{tikz-cd}
\usepackage{bm}
\usepackage{mathtools}
\usepackage{comment}
\usepackage{youngtab}
\pdfoutput=1

\usepackage[norelsize]{algorithm2e}
\usepackage{scalerel}
\usepackage{ulem}
\usepackage{float}
\restylefloat{figure}
\usepackage{stackengine}
\usepackage[english]{babel}
\usepackage[T1]{fontenc}
\usepackage{comment}
\usepackage{array}
\usepackage{lmodern}
\usepackage[mathscr]{euscript}
\usepackage[utf8x]{inputenc}
\usepackage{breqn}
\usepackage{todonotes}
\usepackage{slashed}
\usepackage{youngtab}
\usepackage{amsmath}
\usepackage{amssymb}
\usepackage{amsxtra}
\usepackage{color}
\usepackage[margin=1.2in]{geometry}

\newtheorem{tm}{Theorem}
\newtheorem{pr}[tm]{Proposition}
\newtheorem{cor}[tm]{Corollary}

\theoremstyle{definition}

\newcommand {\CalM} {\mathcal M}
\newcommand {\BR}   {\mathbb R}

\usepackage{color}
\usepackage{graphicx}
\usepackage{epstopdf}

\theoremstyle{plain}
\newtheorem{theorem}{Theorem}

\theoremstyle{remark}
\newtheorem{remark}[theorem]{Remark}

\theoremstyle{definition}

\def\Beq#1#2\Eeq{%
        \begin{equation}%
        \label{#1}%
            #2%
        \end{equation}%
    }

\title{Harmonic locus and Calogero-Moser spaces}

\author{Giovanni Felder}
\address{Department of Mathematics,
ETH Zurich, Switzerland}
\email{giovanni.felder@math.ethz.ch}

\author{Alexander P. Veselov}
\address{Department of Mathematical Sciences,
Loughborough University, Loughborough, UK}
\email{A.P.Veselov@lboro.ac.uk}

\makeatletter
\let\@wraptoccontribs\wraptoccontribs
\makeatother

\contrib[with an appendix by]{Nikita Nekrasov}
\address{Simons Center for Geometry and Physics,
Yang Institute for Theoretical Physics,
Stony Brook University, Stony Brook NY 11794-3636, USA}

\email{nnekrasov@scgp.stonybrook.edu}

\begin{document}

\maketitle

\begin{abstract}
We study the harmonic locus consisting of the monodromy-free Schr\"odinger operators with rational potential and quadratic growth at infinity. It is known after Oblomkov that it can be identified with the set of all partitions via the Wronskian map for Hermite polynomials.
We show that the harmonic locus can also be identified with the subset of the Calogero--Moser space introduced by Wilson, which is fixed by the symplectic action of $\mathbb C^\times.$ As a corollary, for the multiplicity-free part of the locus we effectively solve the inverse problem for the Wronskian map by describing the partition in terms of the spectrum of the corresponding Moser matrix. We also compute the characters of the $\mathbb C^\times$-action at the fixed points, proving, in particular, a conjecture of Conti and Masoero.
In the Appendix written by N. Nekrasov there is an alternative proof of this result, based on the space of instantons and ADHM construction.
 \end{abstract}

\section{Introduction}
One of the most important open problems in quantum integrable systems
is the classification of Schr\"odinger operators with trivial
monodromy.  The corresponding singular set is called locus
configuration.

In dimension one we have the Schr\"odinger equation with a meromorphic
potential $V(z), \, z\in \mathbb C$:
\begin{equation}
\label{static}{}
(-D^2+V(z)) \psi = \lambda \psi, \quad D= \frac{d}{dz}.
\end{equation}
We say that such operator has {\it trivial monodromy} if all the solutions $\psi$ of the corresponding equation \eqref{static} are meromorphic in $z \in \mathbb C$ for {\it all} $\lambda$.
The general problem is to describe all such potentials.

The first results in this direction were found by
Duistermaat and Gr\"unbaum  \cite{DG}, who solved the problem in the class of rational potentials decaying at infinity
$$ 
V=\sum_{i=1}^{n}\frac{m_i(m_i+1)}{(z-z_i)^2}.
$$
They have shown that the corresponding parameters $m_i$ must be integer and that all such potentials are the results of Darboux transformations applied to the zero potential. The potentials  are therefore are given by the Burchnall--Chaundy (or Adler--Moser) explicit formulas \cite{AM, BCh}. 
The corresponding configurations of poles $z_i$ are very special: in the case when all the parameters $m_i=1$ they are none other than the (complex) equilibriums of the Calogero--Moser system with the Hamiltonian
$$
H=\frac12\sum_{i=1}^{n} p_i^2 +
\sum_{i \neq j}^{n}\frac{1}{(z_i-z_j)^2}
$$
and are described by the following algebraic system
$$
\label{eCM}
 \sum_{j\ne i}^{n}\frac{ 1}{(z_i - z_j)^3}=0,\quad  i=1,\dots, {n}.
$$
A remarkable fact discovered first by Airault, McKean and Moser \cite{AMM}) is that this system has no solutions unless the number of particles ${n} = \frac{m(m+1)}{2}$ is a triangular number, in which case the solutions depend on $m$ arbitrary complex parameters. 

Oblomkov \cite{Oblomkov} generalised the Duistermaat--Gr\"unbaum result to the harmonic case, when the rational potentials have quadratic growth at infinity
$$ 
\label{harmonic}
V=z^2 + \sum_{i=1}^{n}\frac{m_i(m_i+1)}{(z-z_i)^2}, \, m_i \in \mathbb Z.
$$ 
He proved that all such potentials can be found by applying Darboux transformations to the harmonic oscillator and explicitly described them via the Wronskians of the Hermite polynomials
$$
u(z)=z^2-2D^2\log W(H_{k_1},\dots, H_{k_l}), \,\, k_1>k_2>\dots >k_l>1.
$$
In \cite{FHV12} we studied the geometry of pole configuration of these potentials in relation with the Young diagrams of the corresponding partitions $$\lambda=(\lambda_1,\dots, \lambda_l), \,\, \lambda_j=k_j-l+j.$$
We were motivated by the following natural question: 
 {\it given the pole set of the potential $u(z)$ from the harmonic locus, how can one find the corresponding partition $\lambda$?}
 
For the so-called doubled partitions we observed numerically that the Young diagram can be ``seen" directly from the shape of the pole set (see Fig. \ref{f-1} below).

In this paper we  show that the harmonic locus can also be identified with a subset of the Calogero--Moser space introduced by Wilson \cite{Wilson}, namely with the
subset fixed by a natural symplectic action of $\mathbb C^\times=\mathbb C\setminus 0.$ As a corollary, for the multiplicity-free part of the locus we answer the question above by describing the partition $\lambda$ explicitly in terms of the spectrum of the corresponding Moser matrix $M$ (see Theorem 1 below).

In the last section we compute the characters of the $\mathbb C^\times$-action at the fixed points. As a corollary, we prove a conjecture by Conti and Masoero \cite{CM} about spectrum of the Hessian matrix $K(\lambda)$ of the Calogero--Moser potential 
$$
K_{ij}(\lambda)=\delta_{ij}\left(1+\sum_{l\neq j}\frac{6}{(z_l-z_j)^4}\right)-(1-\delta_{ij})\frac{6}{(z_i-z_j)^4}, 
$$
where $z_i=z_i(\lambda)$ are the roots of the corresponding Hermite Wronskian $W_\lambda(z)$, which are assumed to be simple.
We prove that in such case the eigenvalues of $K(\lambda)$ have the form
$$
\mathit{Spec} \, K(\lambda)=\{(\lambda_{l(\Box)+1}-c(\Box))^2,\quad \Box \in \lambda\},
$$
where $l(\Box)$ is the leg length of $\Box \in \lambda$ (the number of
boxes of $\lambda$ below $\Box$). This statement is equivalent to
Conjecture 6.3 from \cite{CM}.  For the roots of the usual Hermite
polynomial $H_n(z)$, corresponding to one-row Young diagram, this
agrees with the well-known result that the frequencies of the small
oscillations of Calogero--Moser systems near the corresponding
equilibrium are $1, 2,\dots, n$ (see \cite{Per2, ABCOP}).

An alternative proof of this result based on the space of instantons and ADHM construction is given in the Appendix to this paper, written by N. Nekrasov. 
This implies, in particular, that the formula for the spectrum of $K(\lambda)$ can be rewritten in a more natural way as the set of 
  squares of the hook lengths of $\lambda$:
$
\mathit{Spec} \, K(\lambda)=\{h(\Box)^2,\,\, \Box \in \lambda\}.
$

\section{The harmonic locus and partitions}

By the {\it harmonic locus} $\mathcal{HL}$ we mean the set of the potentials 
\Beq{HL}{}
    u(z) =z^2+ \sum_{i=1}^{n} \frac{m_i(m_i + 1)}{(z - z_i)^2}, \quad z\in \mathbb C,
\Eeq
of the Schr\"odinger operator $L=-D^2+u(z),$ having trivial monodromy in the complex domain, see \cite{DG,V}.
The terminology is inspired by the work of Airault, McKean and Moser \cite{AMM}. The parameters $m_i$ here must be positive integers called multiplicities. The harmonic locus decomposes as a disjoint union $\mathcal{HL}=\cup_{n=0}^\infty\mathcal{HL}_{n}$ according to the number of poles.

Oblomkov \cite{Oblomkov} proved that all corresponding potentials have the form
\[
u(z)=z^2-2D^2\log W(H_{k_1},\dots, H_{k_l}),
\]
where $k_1>k_2>\dots >k_l$ is a
sequence of different positive integers and 
$W (H_{k_1},\dots, H_{k_l})$ is the Wronskian $\det((D^{i-1}H_{k_j})_{i,j})$
 of the corresponding Hermite polynomials $H_k(z)=(-1)^ke^{z^2}D^ke^{-z^2}$:
\[
  H_0(z)=1,\,\,H_1(z)=2z,\,\, H_2(z)=4z^2-2, \,\, H_3(z)=8z^3-12 z, \dots
\]
Moreover different sequences correspond to different potentials.

It is convenient \cite{FHV12} to label these Wronskians $W=W_\lambda(z)$ by the partitions $\lambda=(\lambda_1, \dots, \lambda_l),\, \lambda_1\geq\dots \geq \lambda_l\geq 1,$ such that 
$$k_1=\lambda_1+l-1,\, k_2=\lambda_{2}+l-2,\, \dots,k_{l-1}=\lambda_{l-1}+1,\, k_l=\lambda_l.$$

In \cite{FHV12} it was shown that the Wronskians $W_{\lambda}$ have the following properties:

\medskip

{\it  1. $W_{\lambda}(z)$ is a polynomial in $z$ of degree $|\lambda|=\lambda_1+\lambda_2+\dots +\lambda_l,$
 
 2. $W_{\lambda}(-z)=(-1)^{|\lambda|} W_{\lambda}(z),$
 
 3. $W_{\lambda^*}(z) = (-i)^{|\lambda|} W_{\lambda}(iz),$ where $\lambda^*$ is the conjugate of $\lambda$.}
 
 \medskip
 
 Let $\mathcal P$ be the set of all partitions which we also identify with
 the set of all Young diagrams. By Oblomkov's result
 we have a well-defined bijection $\mathcal W\colon \mathcal P \to \mathcal{HL}$,
 \[
   \mathcal W\colon \lambda \in \mathcal P \to u_\lambda(z)=z^2-2D^2\log W_\lambda(z) \in \mathcal{HL}.
\]

We have the following natural question. Given the potential $u(z)$ from the harmonic locus, how can one find the partition $\lambda$, such that $u(z)=u_\lambda(z)$?
In other words, how to invert the map $\mathcal W$?

For the locus potentials \eqref{HL} with all the multiplicities $m_i$ equal to 1
\[ 
    u(z) =z^2+ \sum_{i=1}^{n} \frac{2}{(z - z_i)^2}
\]
(which we will call {\it simple locus potentials}), we have the following answer, conjectured by the second author in 2012. 

Recall that the {\it content} $c(\Box)$ of the box $\Box=(i,j)$ from the Young diagram $\lambda$ is defined as $j-i$, see Fig.~\ref{f-0}.
It is easy to see that the multiset of contents $$C(\lambda):=\{c(\Box), \,\, \Box \in \lambda\}$$ determines the partition $\lambda$ uniquely.

\begin{tm}\label{t-1} For a simple locus potential $u(z)$, the corresponding
  partition $\lambda$ can be uniquely characterized by the property
  that the contents of $\lambda$ coincide with the eigenvalues of
  Moser's matrix $M$ \[ 
    C(\lambda)= \mathit{Spec} \,M, \quad M_{ij} =
        \begin{cases}
        - \frac{1}{(z_i - z_j)^2} & i \neq j\\
        \sum_{k \neq j}^n \frac{1}{(z_k - z_j)^2} & i=j.
        \end{cases}
\]
\end{tm}

One of the motivations for this answer came from the results of the paper \cite{ABCOP}, where it was shown that when $z_1,\dots, z_{n}$ are the zeros of the Hermite polynomial $H_{n}(z)$ the matrix $M$ has the eigenvalues $0,1,2,\dots, {n}-1,$ which is the content set of the one-row Young diagram. (Calogero conjectured that the converse is also true for the matrices of this form, but he later proved that this does not hold for ${n}=4$, see \cite{C}.)

Another motivation came from the following result of Perelomov \cite{Per}, which explains why we should expect the spectrum of $M$ to be integer.

It is known after Duistermaat and Gr\"unbaum \cite{DG} that the poles $z_1, \dots, z_{n}$ of such potentials satisfy the following {\it locus conditions} (see \cite{V}):
\[ 
\sum_{j \neq i}^{n} \frac{2}{(z_i - z_j)^3} -z_i = 0, \,\, i=1,\dots,{n},
\]
which are necessary and sufficient conditions for trivial monodromy. Note that these conditions are simply the equilibrium conditions 
\Beq{equil}{}
\frac{\partial}{\partial z_i} U(z)=0, \, i=1,\dots, {n} 
\Eeq
for the Calogero--Moser system with the Hamiltonian
\Beq{CM}{}
H=\frac{1}{2}\sum_{i=1}^{n} p_i^2+U(q), \quad U(q)=\frac{1}{2}\sum_{i=1}^{n}q_i^2+\sum_{1\leq i<j\leq {n}} \frac{1}{(q_i - q_j)^2}.
\Eeq

Following \cite{M,Per} introduce matrix $L$ with
   \[ 
   L_{ij}=\frac{1-\delta_{ij}}{q_i - q_j} 
 \]
 and the matrices
 \[ 
L^{\pm}:=L\pm Q, \quad Q_{ij}=q_i \delta_{ij}.
\]

\begin{pr} (Perelomov \cite{Per, Per2})
The equilibrium conditions \eqref{equil} are equivalent to the matrix relations
 \[ 
 [M,L^\pm] = \pm L^\pm.
 \]
 \end{pr}
 
 The proof is a modification of the direct check by Moser \cite{M}.

Thus matrices $L^\pm$ can be viewed as raising/lowering operators for $M$. Since $Me=0$ for $e=(1,\dots,1)^T$ the spectrum of $M$ is integer provided $e$ is cyclic vector for the action of $L^\pm.$

In the paper we establish a direct link of harmonic locus with the Calogero--Moser spaces introduced by Wilson \cite{Wilson}, and as a corollary we provide a proof of Theorem \ref{t-1}.

\section{The harmonic locus and Calogero--Moser spaces}

The Calogero--Moser spaces were introduced by Wilson \cite{Wilson}, who was inspired by the earlier work by Kazhdan, Kostant and Sternberg on the moment map interpretation of the Calogero--Moser systems \cite{KKS}.

Recall that Moser \cite{M} discovered that the system from \cite{Cal71} describing pairwise interacting particles on the line with the Hamiltonian
\[ 
H=\frac{1}{2}\sum_{i=1}^n p_i^2 +\sum_{i<j}^{n} \frac {\gamma^2}{(q_i-q_j)^2}
\]
(now called {\it rational Calogero--Moser system})
can be rewritten in matrix form as
\[ 
\dot L = [L,M],
\]
where
$$
L= \begin{pmatrix}
p_1  &  \frac{\gamma i}{q_1-q_2} &   \frac{\gamma i}{q_1-q_3}   &  \dots  & \dots &  \frac{\gamma i}{q_1-q_n}   &  \\
-\frac{\gamma i}{q_1-q_2}  &  p_2 &  \frac{\gamma i}{q_2-q_3}  & & \dots  & \frac{\gamma i}{q_2-q_n}   & \\
\dots   & \dots & \dots &\dots &  \dots &  \dots &\\
-\frac{\gamma i}{q_1-q_n}    & -\frac{\gamma i}{q_2-q_n} & -\frac{\gamma i}{q_3-q_n}& \dots & -\frac{\gamma i}{q_{n-1}-q_n}  & p_n
\end{pmatrix}
$$
$$
M= -i\gamma\begin{pmatrix}
a_{11}  &  \frac{1}{(q_1-q_2)^2} &  \frac{1}{(q_1-q_3)^2}   &  \dots  & \dots  &\frac{1}{(q_1-q_n)^2}  &  \\
\frac{1}{(q_1-q_2)^2}  &  a_{22} &   \frac{1}{(q_2-q_3)^2}   & \dots &  \dots & \frac{1}{(q_2-q_n)^2}& \\
 \dots  & \dots & \dots &\dots &  \dots &   \dots & \\
\frac{1}{(q_1-q_n)^2}    & \frac{1}{(q_2-q_n)^2} & \frac{1}{(q_3-q_n)2}& \dots & \frac{1}{(q_{n-1}-q_n)^2}  & a_{nn}
\end{pmatrix}
$$
with $a_{ii}=-\sum_{i\neq j}^n  \frac{1}{(q_i-q_j)^2}.$

Following Wilson \cite{Wilson}, we will consider the case when $\gamma=-i$ with attractive (rather than repulsive) potential and
\[ 
H=\frac{1}{2}\sum_{i=1}^n p_i^2 -\sum_{i<j}^{n} \frac {1}{(q_i-q_j)^2},
\]
when the particles may collide. One can view this also as changing time $t \to it, \, i=\sqrt{-1}.$

Note that the corresponding matrix $L$ satisfies the commutation relation
\[ 
[L,Q]=I-ee^T,
\]
where $Q$ is the diagonal matrix with $Q_{ij}=q_i \delta_{ij}$, $I$ is the identity matrix
and the
vector $e$ has all coordinates equal to 1, so $e^T=(1,1,\dots,1)$. Moreover, if $Q$ has simple spectrum, then this commutation relation determines the form of the off-diagonal elements of $L$ uniquely as $L_{ij}=(q_i-q_j)^{-1}.$

Kazhdan, Kostant and Sternberg \cite{KKS} interpreted this relation as a moment map and described the rational Calogero--Moser system as the corresponding symplectic reduction of the free motion on the Lie algebra, which in the repulsive case is the unitary Lie algebra $\mathfrak u_n$. 

Wilson considered a natural complex generalisation of this procedure and introduced the {\it Calogero--Moser space} $\mathcal C_n$ as the quotient space
\[ 
\mathcal C_n=\{(X,Z,v,w): [X,Z]+I=vw\}/\mathit{GL}_n(\mathbb C),
\]
where $X$ and $Z$ are $n$ by $n$ complex matrices, $v$ and $w$ are $n$-dimensional vector and covector (considered as $n\times 1$ and $1\times n$ matrices respectively). The element  $g\in \mathit{GL}_n$ acts as
\[
  (X,Z,v,w)\mapsto (gXg^{-1},gZg^{-1},gv,wg^{-1}).
\]
Wilson showed that $\mathcal C_n$ is a smooth irreducible affine algebraic variety of
dimension $2n$
with many remarkable properties and can be viewed as the quantisation of the Hilbert scheme of $n$ points in the complex plane.

We introduce now the {\it modified Calogero--Moser space} $\mathcal {CM}_n$ as the quotient
\[ 
\mathcal{CM}_n=\{\Pi=(L,Q,M, v, w)\}/\mathit{GL}_n(\mathbb C),
\]
where $L,Q,M$ are $n$ by $n$ complex matrices, $v$ and $w$ are a vector and covector as before, which satisfy the following relations
\begin{align*}
\mathrm{(I)}&: \quad  [L,Q]=I-vw,\\
\mathrm{(II)}&:  \quad [M,Q]=L,\\
\mathrm{(III)}&: \quad [M,L]=Q,\\
\mathrm{(IV)}&: \quad Mv=0, \quad wM=0.
\end{align*}
The group $\mathit{GL}_n$ acts by conjugation on $L,Q,M$ and on $v$, $w$ as before.
 
 Our main result is the following theorem.

\begin{tm}\label{t-3} 
The modified Calogero--Moser space $\mathcal {CM}_n$ is discrete and can be identified with the harmonic locus (and thus with the set of partitions of $n$) via the map $\chi: \mathcal {CM}_n \to \mathcal {HL}_n,$
\[ 
\chi(\Pi)=z^2-2D^2\log \det(zI-Q).
\]
The spectrum of the matrix $M$ is integer and coincides with the content multiset $C(\lambda)$ of the corresponding partition $\lambda$.
 \end{tm}
 
 We start with the following 

\begin{pr}\label{p-4} 
The subset of the modified Calogero--Moser space $\mathcal {CM}_n$ with diagonalisable $Q$ with simple spectrum can be identified with the simple part of the harmonic locus $\mathcal{HL}_n$.
 \end{pr}

\begin{proof} 
 Note first that for diagonal $Q$ with distinct diagonal elements $Q_{ii}=q_i$ Moser's matrices 
\Beq{L}{}
L= \begin{pmatrix}
0  &  \frac{1}{q_1-q_2} &  \dots  &\dots&  \frac{1}{q_1-q_n}   &  \\
-\frac{1}{q_1-q_2}  &  0  & \dots &\dots & \frac{1}{q_2-q_n}   & \\
\dots   & \dots & \dots &\dots &   \dots &\\
-\frac{1}{q_1-q_n}    & -\frac{1}{q_2-q_n} & \dots & -\frac{1}{q_{n-1}-q_n}  & 0
\end{pmatrix}
 \Eeq
\Beq{M}{}
M= \begin{pmatrix}
a_{11}  &  -\frac{1}{(q_1-q_2)^2} &   \dots  & \dots  &-\frac{1}{(q_1-q_n)^2}  &  \\
-\frac{1}{(q_1-q_2)^2}  &  a_{22} &   \dots &  \dots & -\frac{1}{(q_2-q_n)^2}& \\
 \dots  & \dots & \dots &\dots &  \dots & \\
-\frac{1}{(q_1-q_n)^2}    & -\frac{1}{(q_2-q_n)^2} & \dots & -\frac{1}{(q_{n-1}-q_n)^2}  & a_{nn}
\end{pmatrix}
\Eeq
with $a_{ii}=\sum_{i\neq j}^n  \frac{1}{(q_i-q_j)^2},$
 obviously satisfy the relations (I), (II) and (IV) (with $w=(1,1,\dots,1)=v^T$), while the relation (III) is equivalent to the locus conditions by Perelomov's result.
 
Conversely, for any $\Pi \in \mathcal{CM}_n$ with diagonal $Q$ with simple spectrum $q_1,\dots, q_n$ the relation (I) implies that $L_{ij}=\frac{v_iw_j}{q_i-q_j}, \, i\neq j$ and that $v_i w_i=1$ for all $ i=1,\dots, n,$ while the relation (II) implies that $L_{ii}=0$. Conjugation by the diagonal matrix $g\in \mathit{GL}_n(\mathbb C)$ with diagonal elements $v_1,\dots, v_n$ reduces $w$ to $w=e=(1,\dots,1)=v^T$ and $L$ to form \eqref{L}. Now the relations (II) and (IV) imply that $M$ also has Moser's form \eqref{M}, and relation (III) implies the locus condition.
 \end{proof}
 
 Let us now return  to Wilson's Calogero--Moser space $\mathcal C_n$ and consider the natural action of $\mathbb C^\times=\mathbb C\setminus {0}$ on it defined by 
 \Beq{act}{}
 X\mapsto \mu X, \,\, Z \mapsto \mu^{-1} Z,\,\,v\mapsto v, \,\, w\mapsto w ,\, \,\, \mu \in \mathbb C^\times.
 \Eeq
 Let $\mathcal C_n^{\mathbb C^\times}$ be the fixed point subset of $\mathcal C_n$ under the action \eqref{act}.

\begin{pr}\label{p-5} 
  The modified Calogero--Moser space $\mathcal {CM}_n$ can be identified with
  $\mathcal C_n^{\mathbb C^\times}$
via the map
$\nu: (L,Q,M,v,w) \in \mathcal {CM}_n \mapsto (X,Z,v,w) \in \mathcal C_n$ such that
\[ 
X=\frac{1}{2}(L+Q), \quad Z=L-Q.
\]
\end{pr}

\begin{proof}
Since $[X,Z]+I=-[L,Q]+I=vw$,  $(X,Z,v,w)$ indeed belongs to the Calogero--Moser space.
We have also due to relations (II) and (III) that
\Beq{relM}{}
[M,X]=X, \quad
[M,Z]=-Z.
\Eeq
They imply that
$$
e^{tM}Xe^{-tM}=e^{t}X, \quad  e^{tM}Ze^{-tM}=e^{-t}Z,
$$
so the pair $(e^{t}X, e^{-t} Z)$ is $\mathit{GL}_n$-equivalent to $(X,Z)$ and thus $(X,Z)$ is a fixed point of the action \eqref{act}.

Conversely, assume that  $(e^{t}X, e^{-t} Z,v,w)$ is $\mathit{GL}_n$-equivalent to $(X,Z,v,w)$:
$$
e^{t}X=g(t)Xg(t)^{-1}, \,\, e^{-t}Z=g(t)Zg(t)^{-1}, \,\, g(t) \in \mathit{GL}_n(\mathbb C), g(0)=I,
$$
and $g(t)v=v,wg(t)^{-1}=w$.
Differentiating this at $t=0$ and defining $M:=\dot g(0)$ we see that $M$ satisfies the relations \eqref{relM} and hence the relations (II) and (III) for 
\Beq{LQ}{}
L=X+\frac{1}{2}Z, \quad Q=X-\frac{1}{2}Z.
\Eeq
Similarly, differentiating the relations $g(t)v=v, \, wg(t)^{-1}=w$ at $t=0$ we have the relations $Mv=0, wM=0.$
\end{proof}

Luckily, the fixed point set $\mathcal {C}_n^{\mathbb C^\times}$ was studied in detail by Wilson \cite{Wilson}, who identified it with the set $\mathcal P_n$ of all partitions of $n.$ 
Moreover, he explicitly described the bijective map $\kappa\colon \mathcal P_n \to \mathcal {C}_n^{\mathbb C^\times}$ sending $\lambda$ to
\begin{equation}\label{e-kappa}
  \kappa(\lambda)=(X_\lambda,Z_\lambda,v_\lambda,w_\lambda)
\end{equation}
constructed as follows.

Let $\lambda=(\lambda_1,\dots,\lambda_l), \lambda_1+\dots+\lambda_l=n$ be a partition from $\mathcal P_n$.
Following Wilson, we will use an alternative way to describe a partition $$\lambda=(a_1,\dots, a_k | l_1,\dots l_k)$$
introduced by Frobenius \cite{Frob}. Namely,  $k$ is the number of the diagonal boxes in the corresponding Young diagram  with $a_1,\dots, a_k$ and $l_1,\dots, l_k$ (called Frobenius coordinates) being the lengths of the corresponding ``arms" and ``legs" respectively, see Fig.~\ref{f-0}. Let $n_i=a_i+l_i+1$ be the length of the corresponding hook and $r_i=l_i+1, \, i=1,\dots, k.$ Clearly, we have $n_1+\dots +n_k=n=|\lambda|.$

\begin{figure}
\begin{tikzpicture}[scale=.7]
  \foreach\l/\m/\j in {1/0/0, 3/0/1, 4/0/2}     
  {
    \draw (\m,\j+1)--(\l,\j+1);
    \draw (\m,\j)--(\l,\j);
    \foreach\i in {\m,...,\l}
    \draw(\i,\j)--(\i,\j+1);
  }
  \foreach\i/\j in {}   
  \filldraw[gray](\i-1,\j-1)--(\i,\j-1)--(\i,\j)--(\i-1,\j)--cycle;
  \foreach\i/\j/\Label in {1/3/0,2/2/0,1/1/-2,1/2/-1,2/3/1,3/2/1,3/3/2,4/3/3}
  \draw (\i-.5,\j-.5) node {\Label};
\end{tikzpicture}
\caption{The Young diagram of the partition $\lambda=(4,3,1)$ of $|\lambda|=8$ with boxes labeled by their contents.
  Its Frobenius representation is $\lambda=(3,1|2,0)$ and the multiset of contents
  is  $C(\lambda)=\{-2,-1,0,0,1,1,2,3\}$.}
\label{f-0}
\end{figure}

Let
$$
\Lambda_m= \begin{pmatrix}
0  &  1 &    &    &    &  \\
0  &  0 &  1  & & & \\
         &  0 &  0  & 1  &  & \\
   & & \ddots &\ddots &  \ddots & \\
   & & & 0 & 0 & 1\\
   & & & & 0 & 0
\end{pmatrix} 
$$
be the standard $m\times m$ Jordan block, then $Z_\lambda$ is defined simply as the block-diagonal matrix with Jordan blocks 
$\Lambda_{n_1}, \dots, \Lambda_{n_k}$ on the diagonal. The block $X_{ij}$ of the matrix $X_\lambda$ is the unique $n_i$ by $n_j$ matrix with non-zero entries in
the diagonal number $r_j-r_i-1$\footnote{The diagonal number $m$ of a matrix consists of the entries indexed by $i,j$ such
  that $j-i=m$}
obeying the commutation relation
$$
X_{ij}\Lambda_{n_j}-\Lambda_{n_i}X_{ij}=n_iE(r_i, r_j)-\delta_{ij}I,
$$
where $E(r,s)$ is the matrix with the only non-zero element $E(r,s)_{r,s}=1$. In particular, the block $X_{ii}$ has the form
$$
X_{ii}= \begin{pmatrix}
0  &  0 &    &    &    &  \\
1  &  0 &  0  & & & \\
         &  2 &  0  & 0  &  & \\
   & & \ddots &\ddots &  \ddots & \\
   & & & -2 & 0 & 0\\
   & & & & -1& 0
\end{pmatrix}
$$
with the entries $1, 2, \dots, l_i-1, l_i;  -a_i, -a_i+1, \dots, -2, -1$ on the diagonal number $-1$
(see \cite{Wilson}). Finally $v_\lambda$ and $w_\lambda$ are such that
$v_\lambda w_\lambda$ has blocks $n_iE(r_i,r_j)$.

Using this, we can define the matrices $L_\lambda$ and $Q_\lambda$ by \eqref{LQ}, so we only need to find the corresponding $M$.

\begin{pr}\label{p-6} 
In Wilson's basis the corresponding matrix $M_\lambda$ is diagonal with blocks 
\[ 
M_{ii}={\textit diag}\, (-l_i, -l_i+1, \dots,-2, -1, 0, 1, 2, \dots, a_i-1, a_i).
\]
\end{pr}

The proof follows from direct check of the relations (II), (III), (IV).

\begin{cor}\label{c-7} 
The spectrum of $M$ coincides with the content multiset of the partition $\lambda:$
$
\mathit{Spec} \, M = C(\lambda).
$
\end{cor}
Indeed, the eigenvalues of $M_{ii}$ are simply the contents of the boxes in the $i$-th diagonal hook of the Young diagram of $\lambda.$

\begin{pr}\label{p-8} 
The characteristic polynomial of the matrix $Q_\lambda=X_\lambda-\frac{1}{2}Z_\lambda$ coincides with the Hermite Wronskian
\[ 
\det (zI-Q_\lambda)=A(\lambda)W_\lambda(z)
\]
for some constant $A(\lambda).$
\end{pr}
 
\begin{proof}
Firstly, we use the result from Wilson's paper \cite{Wilson} (see formula (6.14) on page 29), which claims that for some constant $B(\lambda)$
\[ 
\det(X_\lambda-\sum_{i\geq 1}p_i(-Z_\lambda)^{i-1})=B(\lambda) s_\lambda,
\]
where $p_i$ are the power sum symmetric functions and $s_\lambda$ is the corresponding Schur symmetric function \cite{Mac}.

Then we use the following general fact expressing the Wronskians of Appell polynomials in terms of Schur symmetric functions, see \cite{BHSS}.

Recall that the sequence of polynomials $A_k(x), k\geq 0$ is called {\it Appell} if they satisfy the relation
$$
\frac{d}{dx}A_k(x)=kA_{k-1}(x), \quad k\geq 0
$$
with $A_0=1.$
Let $F_A(t)=e^{xt}f_A(t), \,\, f_A(t)=1+\sum_{k\geq 1}a_k\frac{t^k}{k}$ be the exponential generating function of the Appell sequence
$$
F_A(t)=e^{xt}f_A(t)=\sum_{k\geq 0}A_k(x)\frac{t^k}{k!}.
$$
Consider the expansion
$$
\log F_A(t)=xt+\log f_A(t)=xt+\sum_{k\geq 1}b_k\frac{t^k}{k},
$$
then from \cite{BHSS}, Theorem 4.1 and Proposition 4.3 it follows that the Wronskian 
$$
A_\lambda(x)=W(A_{k_1}(x), \dots, A_{k_l}(x)), \quad k_i=\lambda_i+l-i,
$$
coincides with the Schur symmetric function $s_\lambda(p_1,p_2,\dots)$ if we specialize $$p_1=x+b_1, \, p_i=b_i, \, i\geq 2.$$

The Hermite polynomials (in physicists' version) satisfy the relation
$$
\frac{d}{dx}H_k(x)=2kH_{k-1}(x), \quad k\geq 0,
$$
so the scaled monic versions $\tilde H_k(x):=2^{-k}H_k(x)$ form an Appell sequence. The corresponding generating function is 
$
F_H(t)=e^{xt-\frac{t^2}{4}}
$
(see e.g. Szeg\H o \cite{Sz}).
Since $\log F_H(t)=xt-\frac{1}{4}t^2$ we have $b_1=0, \, b_2=-\frac{1}{2}$ with $b_i=0, i>2.$

This implies that the Hermite Wronskian $W_\lambda(z)$ up to a constant multiple coincides with the Schur function $s_\lambda(z, -\frac{1}{2}, 0, 0, \dots,)$
so, by Wilson's result, up to a constant multiple
\[ 
W_\lambda(z)=\det(X_\lambda-zI-\frac{1}{2} Z_\lambda)=\det(Q_\lambda-zI),
\]
as claimed.
\end{proof}

We can now prove our Theorem \ref{t-3}.
We need to show that the map
\[
  \chi\colon (L,Q,M,v,w)\mapsto u(z)=z^2-2D^2\log\det(zI-Q)
\]
is
a bijection from $\mathcal {CM}_n$ to $\mathcal{HL}_n$.
We have a diagram of maps
\[
  \begin{tikzcd}
    \mathcal{CM}_n \arrow[r,"\nu"]\arrow[d,"\chi"]
    &
  \mathcal C_n^{\mathbb C^\times}\arrow[dl,"\gamma"]
  \\
  \mathcal{HL}_n &
  \mathcal P_n\arrow[l,"\mu"]\arrow[u,"\kappa"]
\end{tikzcd}
\]
In this diagram the maps $\kappa$, $\mu$, $\nu$ are bijections due to
Wilson, Oblomkov, and Proposition \ref{p-5}. The map 
$$\gamma: (X,Z,v,w) \to z^2-2D^2\log\det(zI-X+\frac12 Z)$$ makes
the lower triangle commutative by Proposition \ref{p-8}. By construction,
$\chi=\gamma\circ\nu$ and is thus a bijection, completing the proof
of Theorem \ref{t-3}.

Theorem \ref{t-1} follows then from Proposition \ref{p-4}.

\section{Characters of $\mathbb C^\times$-action and the Conti--Masoero conjecture}

Consider again the simple part of the harmonic locus with the potential
$$
u(z) =z^2+ \sum_{i=1}^{n} \frac{2}{(z - z_i)^2}
$$
satisfying the locus conditions
$$
\sum_{j \neq i}^{n} \frac{2}{(z_i - z_j)^3} -z_i = 0, \,\, i=1,\dots,{n},
$$
which can be interpreted as the equilibrium conditions for the Calogero--Moser system with the potential
$$
 U(q)=\frac{1}{2}\sum_{i=1}^{n}q_i^2+\sum_{1\leq i<j\leq {n}} \frac{1}{(q_i - q_j)^2}.
$$
Consider the (complex) Hessian of the potential $U(q)$ at the corresponding equilibrium
\[ 
K_{ij}:=\left.\frac{\partial ^2 U}{\partial q_i \partial q_j}\right |_{q=z}=\delta_{ij}\left(1+\sum_{l\neq j}\frac{6}{(z_l-z_j)^4}\right)-(1-\delta_{ij})\frac{6}{(z_i-z_j)^4}.
\]

Since all the solutions of Calogero--Moser system \eqref{CM} are known to be periodic with period $2\pi$ (see \cite{OP}), the eigenvalues of $K$ must be the squares of some integers (cf. \cite{Per2}). Indeed, for the roots of Hermite polynomial $H_{n}(z)$ the corresponding eigenvalues are known to be $1, 2^2, 3^2, \dots, {n}^2$, as follows from the
 formula
$$
K=(M+I)^2,
$$
relating $K$ to Moser's matrix $M$ (see \cite{Per, Per2} and Proposition 3.4 in \cite{ABCOP}).

However, already for $\lambda=(3,1)$ one can check that the corresponding matrix $K$ does not commute with $M$, which means that we can not apply our Theorem 3 to compute the spectrum of $K$ in the general case.

We will use instead our Proposition 5 and Wilson's results \cite{Wilson} to compute these eigenvalues, proving, in particular, a conjectural formula by Conti and Masoero (see Conjecture 6.3 in \cite{CM}). 

More generally, consider the action of $\mathbb C^\times$ on the Calogero--Moser space $\mathcal C_{n}$ defined by \eqref{act}. We have identified the fixed points of this action with the harmonic locus. Let $(X_\lambda, Y_\lambda, v_\lambda, w_\lambda)$ be the fixed point corresponding to partition $\lambda\in \mathcal P_{n}$ and consider the character of the linearised action of $\mathbb C^\times$ at this point:
\[
\chi_\lambda(\mu):=\sum_{s\, \in S_\lambda} \mu^s, \quad \mu \in \mathbb C^\times.
\]
Note that because the action of $\mathbb C^\times$ is symplectic, the weight set $S_\lambda$ of the action  is invariant under the change $s \to -s:$
\[ 
S_\lambda=\{\pm s_1(\lambda), \dots, \pm s_{n}(\lambda)\}, \quad s_k(\lambda) \in \mathbb Z_{>0}.
\]

\begin{tm}\label{t-4} 
The character of the linearised $\mathbb C^\times$-action at the fixed point $(X_\lambda, Y_\lambda, v_\lambda, w_\lambda)\in C_{n}$ has the form
\Beq{charform}
\chi_\lambda(\mu)=(\mu-2+\mu^{-1})G_\lambda(\mu)G_\lambda(\mu^{-1})+G_\lambda(\mu)+G_\lambda(\mu^{-1}),
\Eeq
where 
\[ 
G_\lambda(\mu):=\sum_{\Box \in \lambda} \mu^{c(\Box)}.
\]
with $c(\Box)$ being, as before, the content of the $\Box$ in the Young diagram of $\lambda$.
 \end{tm}

\begin{proof}
  Denote by $\mathit{Mat}_{k,l}$ the space of complex $k$ by $l$ matrices.  The
  tangent space at a point $(X,Z,v,w)\in \mathcal C_n$ is given by the
  linearized symplectic reduction, namely as the middle cohomology
  $\operatorname{Ker}(\delta')/\operatorname{Im}(\delta)$ of the
  complex
  \begin{equation}\label{e-cx}{}
    \mathit{Mat}_{n,n}
    \stackrel{\delta}
    \to \mathit{Mat}_{n,n}\oplus \mathit{Mat}_{n,n}\oplus
    \mathit{Mat}_{n,1}\oplus \mathit{Mat}_{1,n}
    \stackrel{\delta'}\to \mathit{Mat}_{n,n}.
  \end{equation}
Here $\delta'$ is the differential of the  moment map
\[
  \delta'\colon (\xi,\zeta, \rho,\sigma)\mapsto [\xi,Z]+[X,\zeta]+\rho w+v \eta
\]
and $\delta$ is the infinitesimal action of $\mathit{GL}_n$
\[
  \delta\colon \alpha\mapsto ([\alpha,X],[\alpha,Z],\alpha v,-w\alpha).
\]
The $\mathit{GL}_n$-orbit of
$\kappa(\lambda)=(X_\lambda,Z_\lambda,v_\lambda,w_\lambda)$
is fixed by the $\mathbb C^\times$-action \eqref{act}. By the relations
(II)--(IV) and Proposition \ref{p-5},
$\kappa(\lambda)$ itself
is a fixed point for the twisted action of $\mathbb C^\times$
\begin{equation}\label{e-act}{}
  (X,Z,v,w)\mapsto (\mu^{1-M}X\mu^{M},\mu^{-1-M}Z\mu^M,\mu^{-M}v,w\mu^M),
\quad  \mu\in\mathbb C^\times.
\end{equation}
The maps $\delta,\delta'$ are $\mathbb C^\times$-equivariant for this
action on the middle term of \eqref{e-cx} and for the action
$\alpha\mapsto \mu^{-M}\alpha\mu^{M}$, $\alpha\in \mathit{Mat}_{n,n}$,
$\mu\in \mathbb C^\times$ on the extreme terms.

Since $\delta$ is injective and $\delta'$ is surjective (see \cite{Wilson},
Corollary 1.4 and Proposition 1.7),  the character of the
$\mathbb C^\times$-action on the cohomology is equal to the
character-valued Euler characteristic of the complex, namely as the
character of the middle term minus the sum of characters of the
extreme terms.

By Corollary \ref{c-7}  the eigenvalues of $M$ are contents of $\lambda$ and
the eigenvalues of $[M,-]$ on $\mathit{Mat}_{n,n}$ are differences of
contents. Thus the character of the middle term for the
$\mathbb C^\times$-action \eqref{e-act} is
\[
  \mu
  G_{\lambda}(\mu^{-1})G_\lambda(\mu)+\mu^{-1}G_\lambda(\mu^{-1})G_\lambda(\mu)+G_\lambda(\mu^{-1})+G_\lambda(\mu).
\]
From this we subtract twice the character
$G_\lambda(\mu^{-1})G_\lambda(\mu)$ from the extreme terms and obtain  formula (\ref{charform}). 
\end{proof}

At the simple part of the harmonic locus the set $S_\lambda$ determines the spectrum of the Hessian $K(\lambda)$ at the corresponding equilibrium by the formula
$$
\mathit{Spec} \, K(\lambda)= \{s_1(\lambda)^2,\dots, s_{n}(\lambda)^2\}.
$$
In the paper \cite{CM} Conti and Masoero conjectured a recursive procedure to compute the eigenvalues of $K$, which can be shown to be equivalent to the following formula for the spectrum of $K$:
\Beq{CMconj}
\mathit{Spec} \, K(\lambda)=\{(\lambda_{l(\Box)+1}-c(\Box))^2,\quad \Box \in \lambda\},
\Eeq
where $l(\Box)$ is the leg length of $\Box \in \lambda.$

\begin{pr}\label{p-10} 
For any partition $\lambda \in \mathcal P_{n}$ we have the identity
\Beq{identity}
\chi_\lambda(\mu)=\sum_{\Box \in \lambda} \mu^{\lambda_{l(\Box)+1}-c(\Box)}+\mu^{-\lambda_{l(\Box)+1}+c(\Box)}.
\Eeq
\end{pr}

The proof is by induction in the length of the partition using the recursive procedure from \cite{CM}.

In particular, for the one-hook Young diagram $\lambda=(a|l)$ (in Frobenius notations) we have 
$$
G_{(a|l)}(\mu)=\frac{\mu^{a+1}-\mu^{-l}}{\mu-1},
$$
$$
\chi_{(a|l)}(\mu)=-(\mu^{a+1}-\mu^{-l})(\mu^{-(a+1)}-\mu^{l})+\frac{\mu^{a+1}-\mu^{-l}}{\mu-1}+\frac{\mu^{-(a+1)}-\mu^{l}}{\mu^{-1}-1}
$$
$$
=\sum_{j=-l, \, j\neq 0}^a(\mu^j+\mu^{-j})+\mu^{a+l+1}+\mu^{-(a+l+1)}.
$$

\begin{cor} 
The Conti--Masoero conjecture holds for any simple locus potential $u_\lambda(z)$.
\end{cor}

\begin{remark}
  As it was pointed to us by Nikita Nekrasov \cite{NN}, the set of eigenvalues of $K(\lambda)$ can be written in much more natural way as the set of 
    squares of the  hook lengths of $\lambda$:
\Beq{hooks}
\mathit{Spec} \, K(\lambda)=\{h(\Box)^2,\quad \Box \in \lambda\},
\Eeq
where for the box $\Box=(i,j)\in \lambda$ the hook length $h(\Box)$ is defined by $$h(\Box)=\lambda_i-j+\lambda^*_j-i+1$$ with the Young diagram $\lambda^*$ being conjugate to 
$\lambda.$ Indeed, it is easy to check that we have an equality
\Beq{eq}
\lambda_{l(\Box)+1}-c(\Box)=h(\Box^*)
\Eeq
where for any $\Box=(i,j)\in \lambda$ the box $\Box^*$ is defined by $(i,j)^*=(\lambda^*_j-i+1, j).$
 
 Nekrasov's proof is independent and based on his earlier work on the ADHM construction of instantons (see the Appendix to this paper).
\end{remark}

\section{Concluding remarks}

For the simple part of the harmonic locus we have now an effective way to
invert the Wronskian map $\lambda \to u_\lambda(z)$ by computing the
spectrum of the corresponding Moser matrix $M.$ The question about a
correct analogue of the matrix $M$ when we have multiplicities is still
open. Note that there is a conjecture saying that for the harmonic
locus only $z=0$ may have multiplicity larger than 1 (see
\cite{FHV12}), but in spite of a large numerical evidence and some
recent progress in this direction \cite{Duran, Grosu} the general
proof is still to be found.

In the case of the so-called doubled partitions
$\lambda^{2}=(\lambda_1,\lambda_1, \dots, \lambda_l,\lambda_l)$ we
have observed in \cite{FHV12} a surprising relation between the shape
of the Young diagram of $\lambda$ and the pattern of zeroes of the
corresponding Hermite Wronskian, which allows to ``see" 2 copies of
the corresponding partition $\lambda$ (see Fig. \ref{f-1}).

\begin{figure}[h]
\centerline{ \includegraphics[width=.5\textwidth]{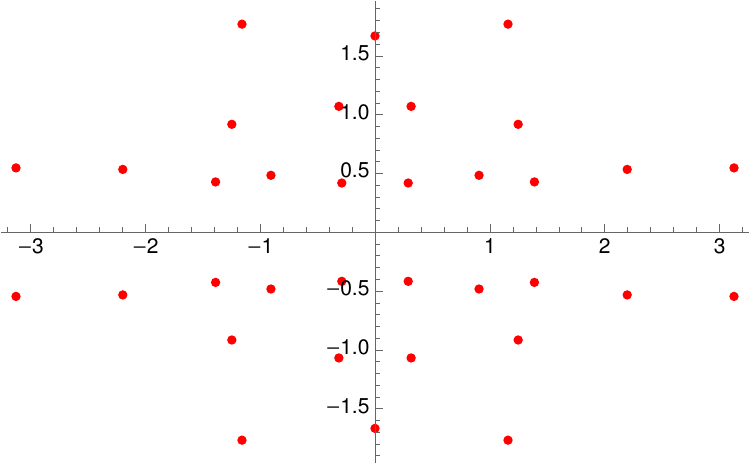} }
\caption{Zeroes of the Wronskian $W_{\lambda^2}$ for the doubled partition $\lambda^{2}$ with $\lambda=(10,4,3)$} \label{compa1}\label{f-1}
\end{figure}

Note that due to the results of M.G. Krein and V.E. Adler \cite{Adler} such partitions can be characterised by the property that the corresponding Hermite Wronskian has no real roots (and thus conjecturally has only simple roots).

The harmonic locus is related to the so-called {\it monster potentials} 
$$
V(x)= \frac{L}{x^2}+x^{2\alpha}-2D^2\sum_{k=1}^{n}\log (x^{2\alpha +2}-z_k)
$$
introduced by Bazhanov, Lukyanov and Zamolodchikov \cite{BLZ}. When
$\alpha=1$ and $L=m(m+1)$ they coincide with the simple locus
potentials $u_\lambda(x)$ corresponding to the symmetric Young
diagrams $\lambda=\lambda^*.$ There are several outstanding
conjectures about the number of monster potentials, which are open
even in this case (see \cite{CM} for some recent results in this
direction). Our results might be useful in this setting as well.

\section{Acknowledgements}

APV is grateful to FIM, ETH Zurich for the hospitality in April 2024 when this work was completed.

\appendix

\newcommand{\nn}[1]{{\color{blue}\bf NN: {#1}}}

\newcommand{\boxit}[1]{\vbox{\hrule\hbox{\vrule\kern8pt
\vbox{\hbox{\kern8pt}\hbox{\vbox{#1}}\hbox{\kern8pt}}
\kern8pt\vrule}\hrule}}
\newcommand{\mathboxit}[1]{\vbox{\hrule\hbox{\vrule\kern8pt\vbox{\kern8pt
\hbox{$\displaystyle #1$}\kern8pt}\kern8pt\vrule}\hrule}}

\newcommand{\picit}[2]{\includegraphics[width=#1cm]{#2.eps}}

\newcommand{\gcr}{{\mathfrak X}} 

 \newcommand{\beq}{\begin{equation}}
                \newcommand{\bea}{\begin{eqnarray}}
                \newcommand{\eea}{\end{eqnarray}}
                 \newcommand{\eeq}{\end{equation}}  

\newcommand{\vev}[1]{\left\langle\ #1 \ \right\rangle}
\newcommand{\nvev}[1]{\biggl\langle\biggl\langle\, #1 \,\biggr\rangle\biggr\rangle}
\newcommand{\my}[1]{{\mathscr Y} \left( #1 \right)}
\newcommand{\y}{{\mathscr Y}}
\newcommand{\q}{{\mathscr Q}}
\newcommand{\x}{{\mathscr X}}
\newcommand{\kzD}{{\mathscr D}^{\qe}}
\newcommand{\mA}{{\mathscr A}}
\newcommand{\mN}{{\mathscr N}}
\newcommand{\mV}{{\mathscr V}}
\newcommand{\mU}{{\mathscr U}}
\newcommand{\cV}{{\mathcal{V}}}
\newcommand{\iM}{{\mathscr M}}
\newcommand{\Lf}{{\mathfrak{L}}}
\newcommand{\Sf}{{\mathfrak{S}}}
\newcommand{\sM}{{\mathfrak{M}}}
\newcommand{\sV}{{\mathfrak{V}}}
\newcommand{\sE}{{\mathfrak{E}}}
\newcommand{\uM}{\overline{\mathcal M}}
\newcommand{\mM}{{\mathfrak M}}
\newcommand{\mW}{{\mathscr W}}
\newcommand {\BA}   {\mathbb A}
\newcommand {\BB}   {\mathbb B}
\newcommand {\BD}   {\mathbb D}
\newcommand {\BC}   {\mathbb C}
\newcommand {\BF}   {\mathbb F}
\newcommand {\BI}   {\mathbb I}
\newcommand {\BO}   {\mathbb O}
\newcommand {\BH}   {\mathbb H}
\newcommand {\BN}   {\mathbb N}
\newcommand {\BL}   {\mathbb L}
\newcommand {\BP}   {\mathbb P}
\newcommand {\BQ}   {\mathbb Q}
\newcommand {\BT}   {\mathbb T}
\newcommand {\bB}   {\mathbf{B}}
\newcommand {\bI}   {\mathbf{I}}
\newcommand {\bN}   {\mathbf{N}}
\newcommand {\bM}   {\mathbf{M}}
\newcommand {\bJ}   {\mathbf{J}}
\newcommand {\bP}   {\mathbf{P}}
\newcommand {\bG}   {\mathbf{G}}
\newcommand {\bQ}   {\mathbf{Q}}
\newcommand {\bR}   {\mathbf{R}}
\newcommand {\bq}   {\mathscr P}
\newcommand {\mv}   {\mathscr M}
\newcommand {\qe} {\mathfrak q}
\newcommand {\pe} {\mathfrak p}
\newcommand {\gf} {\mathfrak g}
\newcommand {\ib} {\bar{i}}
\newcommand {\nf} {\mathfrak n}
\newcommand {\jb} {\bar{j}}
\newcommand {\kb} {\bar{k}}
\newcommand {\Pf} {\mathfrak{P}}
\newcommand {\Xf} {\mathfrak{X}}
\newcommand {\Yf} {\mathfrak{Y}}
\newcommand {\Hf} {\mathsf{H}}
\newcommand {\Hadhm} {\mathsf{H_{ADHM}}}
\newcommand{\Hilb} {\mathsf{Hilb}}
\newcommand{\HM} {\mathsf{HM}}
\newcommand{\Ci} {\mathsf{C}_{\ib, \bv}}
\newcommand {\Jf} {\mathfrak{J}}
\newcommand {\Gp} {\mathsf{G}}
\newcommand {\Gpl} {\mathsf{G}^{\rm loc}}
\newcommand {\hf} {\mathsf{h}}
\newcommand {\xb} {\mathbf{x}}
\newcommand {\ii} {\mathrm{i}}
\newcommand {\bT}   {\mathbf{T}}
\newcommand {\ba}  {\underline{\ac}}
\newcommand {\bb} {\mathbf{b}}
\newcommand {\ex}  {\mathbf{e}}
\newcommand {\Luv}{{\beta_1}}
\newcommand {\bg} {\mathbf{g}}
\newcommand {\bbr}{  \mathbf{r}}
\newcommand {\bn}{\underline{\mathbf{n}}}
\newcommand {\mm}{\underline{\mathbf{m}}}
\newcommand {\mt} {\tt m}
\newcommand {\Det} {\tt Det}
\newcommand {\dt} {\tt d}
\newcommand {\bc} {\underline{\mathbf{c}}}
\newcommand {\bv} {\underline{\mathbf{v}}}
\newcommand {\bk}{  \mathbf{k}}
\newcommand {\bl}{  \mathbf{l}}
\newcommand {\bw} {\underline{\mathbf{w}}}
\newcommand {\vt} {\vartheta}
\newcommand {\wt} {\tt w}
\newcommand {\xt} {\tt x}
\newcommand {\tw} {\text{w}}

\newcommand {\bfc} {\underline{\fc}}
\newcommand {\bnu} {\underline{\boldsymbol{\nu}}}
\newcommand {\bmu} {\boldsymbol{\mu}}
\newcommand {\bmt} {\underline{\boldsymbol{\mt}}}
\newcommand {\bzt} {\underline{\boldsymbol{\zeta}}}
\newcommand {\btt} {\underline{\boldsymbol{\tau}}}
\newcommand {\ept} {\underline{\ec}}
\newcommand {\rv} {\underline{\rho}}
\newcommand {\bzv} {\underline{\mathbf{z}}}
\newcommand {\bpv} {\underline{\mathbf{p}}}
\newcommand {\bqt} {\underline{\qe}}
\newcommand {\Dq} {\mathscr D^{\qe}}
\newcommand {\Nq} {\mathscr\nabla^{\qe}}
\newcommand {\bkt} {\underline{\bk}}
\newcommand {\bnt} {\underline{\bf {\vec n}}}
\newcommand {\blt} {\underline{\bl}}
\newcommand {\bla} {\underline{\boldsymbol{\lambda}}}
\newcommand {\bwt}{ \mathbf{\tilde w}}
\newcommand {\bu}{ \mathbf{u}}
\newcommand {\bU}{ \mathbf{U}}
\newcommand {\bx}{  \mathbf{x}}
\newcommand {\bX}{ \mathbf{X}}
\newcommand {\bz}{  \mathbf{z}}
\newcommand {\xr}{  \mathrm{x}}
\newcommand {\bS}{ \mathbf{S}}
\newcommand {\BS}   {\mathbb S}
\newcommand {\BDe}   {\boldsymbol{\Delta}}
\newcommand {\BW}   {\mathbb W}
\newcommand {\BZ}   {\mathbb Z}

\newcommand {\RP}   {\mathbb R \mathbb P}
\newcommand {\CP}   {\mathbb C \mathbb P}
\newcommand {\WP}   {\mathbb W \mathbb P}

\newcommand {\Unity}{\mathbf 1}
\newcommand{\Cross}{\mathbin{\tikz [x=1.4ex,y=1.4ex,line width=.2ex] \draw (0,0) -- (1,1) (0,1) -- (1,0);}}%

\newcommand {\bsa} {\underline{\boldsymbol{a}}}

\newcommand {\ac} {\mathfrak{a}}
\newcommand {\ec} {\mathfrak{e}}
\newcommand {\fb} {\mathfrak{b}}
\newcommand {\fc} {\mathfrak{c}}
\newcommand {\fC} {\mathfrak{C}}
\newcommand {\tqe}{\tilde{\mathfrak{q}}}
\newcommand {\fe} {\mathfrak{f}}
\newcommand {\fu} {\mathfrak{u}}
\newcommand {\ma} {\mathfrak{m}}
\newcommand {\zb} {{\bar z}}
\newcommand {\yb} {{\bar y}}
\newcommand {\vb} {{\bar v}}
\newcommand {\wb} {{\bar w}}
\newcommand {\CalA} {\mathcal A}
\newcommand {\CalB} {\mathcal B}
\newcommand {\CalC} {\mathcal C}
\newcommand {\obs} {\mathscr C}
\newcommand {\CalD} {\mathcal D}
\newcommand {\CalE} {\mathcal E}
\newcommand {\CalF} {\mathcal F}
\newcommand {\CalG} {\mathcal G}
\newcommand {\CalH} {\mathcal H}
\newcommand {\CalI} {\mathcal I}
\newcommand {\CalJ} {\mathcal J}
\newcommand {\CalK} {\mathcal K}
\newcommand {\CalL} {\mathcal L}
\newcommand {\CalN} {\mathcal N}
\newcommand {\CalO} {\mathcal O}
\newcommand {\CalP} {\mathcal P}
\newcommand {\CalQ} {\mathcal Q}
\newcommand {\CalR} {\mathcal R}
\newcommand {\CalS} {\mathcal S}
\newcommand {\CalT} {\mathcal T}
\newcommand {\CalU} {\mathcal U}
\newcommand {\CalV} {\mathcal V}
\newcommand {\CalX} {\mathcal X}
\newcommand {\CalY} {\mathcal Y}
\newcommand {\CalW} {\mathcal W}
\newcommand {\CalZ} {\mathcal Z}

\newcommand {\ee} {\mathfrak E}
\newcommand {\es} {\mathscr E}

\newcommand {\hs} {\mathscr H}

\newcommand{\al}{\alpha}
\newcommand{\ve}{\varepsilon}
\newcommand{\ep}{\epsilon}
\newcommand{\om}{\omega}
\newcommand{\eu}{\mathrm{eu}}

\newcommand{\gym}{g_{\text{\tiny \textsc{ym}}}}


\renewcommand{\hat}{\widehat}

\newcommand{\Gg}{\mathsf{G}_{\mathbf{g}}}
\newcommand{\fgg}{\mathfrak{g}_{\mathbf{g}}}

\newcommand{\tgc}{\mathfrak{t}^{\BC}_{\mathbf{g}}}

\newcommand{\Gf}{\mathsf{G}_{\text{\tiny f}}}
\newcommand{\Gr}{\mathsf{G}_{\text{\tiny rot}}}
\newcommand{\Tf}{T_{\text{\tiny f}}}
\newcommand{\tf}{\mathfrak{t}_{\text{\tiny f}}}
\newcommand{\tfc}{\mathfrak{t}^{\BC}_{\text{\tiny f}}}

\newcommand{\Gammadi}{\boldsymbol{\Gamma}}

\newcommand{\iw}{{^{i}{\CalW}}}

\newcommand{\Gad}{{\mathbf{G}}^{\text{ad}}}
\newcommand{\Tad}{{\mathbf{T}}^{\text{ad}}}

\newcommand{\gq}{\mathfrak{g}_{\gamma}}
\newcommand{\gqp}{\mathfrak{g}_{\gamma'}}
\newcommand{\hq}{\mathfrak{h}_{\gamma}}

\newcommand{\agq}{\mathfrak{\hat g}_{\gamma}}
\newcommand{\tq}{\mathfrak{t}_{\gamma}}
\newcommand{\tqc}{\mathfrak{t}^{\BC}_{\gamma}}
\newcommand{\sdtimes}{\mathbin{
\hbox{\hskip2pt
\vrule height 4.1pt depth -.3pt width.25pt\hskip-2pt$\times$}}}

\newcommand{\Vg}{\mathsf{V}_{\gamma}}
\newcommand{\Vgp}{\mathsf{V}_{\gamma^{+}}}
\newcommand{\Vgm}{\mathsf{V}_{\gamma^{-}}}
\newcommand{\Vgmp}{\mathsf{V}_{\gamma^{\mp}}}
\newcommand{\Vgpm}{\mathsf{V}_{\gamma^{\pm}}}
\newcommand{\Verf}{\mathrm{Vert}_{\text{f}}}
\newcommand{\Eg}{\mathsf{E}_{\gamma}}
\newcommand{\Egp}{\mathsf{E}_{\gamma^{+}}}
\newcommand{\Egpm}{\mathsf{E}_{\gamma^{\pm}}}
\newcommand{\Egm}{\mathsf{E}_{\gamma^{-}}}
\newcommand{\Arr}{\mathsf{Arrows}_{\gamma}}
\newcommand{\Path}{\mathsf{Paths}_{\gamma}}
\newcommand{\Obs}{\mathsf{Obs}}
\newcommand{\Def}{\mathsf{Def}}
\newcommand{\Edgg}{\mathsf{Edge}_{\gamma, \Gamma}}
\newcommand{\am}{\mathrm{a}}
\newcommand{\ca}{\mathrm{c}}
\newcommand{\Tr}{\mathsf{Tr}\,}
\newcommand{\im}{\mathsf{im}\,}
\newcommand{\dg}{\mathsf{deg}}


\newcommand {\3}{\underline{\bf 3}}
\newcommand {\4}{\underline{\bf 4}}
\newcommand {\6}{\underline{\bf 6}}
\newcommand {\Hfr} {\mathfrak{H}}
\newcommand {\be}{\underline{\mathbf{e}}}
\newcommand {\bE}{\underline{\mathbf{E}}}
\newcommand {\fs} {\mathfrak{s}}
\newcommand {\uA}{\underline{\mathbf{A}}}
\newcommand{\fo}{\vert\kern -.03in\_}
\newcommand {\bht}{\underline{\mathbf{h}}}

\section{Eigenvalues of Hessian of CM hamiltonian, by Nikita
Nekrasov}

\subsection{ADHM equations} 

Consider the moduli space ${\CalM}_{k}$ of charge $k$ $U(1)$ instantons on noncommutative ${\BR}^{4}_{\theta}$ \cite{Nekrasov:1998ss}.  It is defined as the quotient of the space of 
 solutions to the deformed ADHM equations 
 \beq
 \begin{aligned}
 & [ B_{1}, B_{1}^{\dagger} ] + [ B_{2}, B_{2}^{\dagger} ] + I I^{\dagger} - J^{\dagger} J = {\zeta} \cdot {\bf 1}_{K} \\
 & [ B_{1} , B_{2} ] + IJ = 0 \\
 \end{aligned}
 \label{eq:defadhm}
 \eeq
 by the action of the group $U(K)$ acting via:
 \beq
 g \cdot ( B_{1}, B_{2}, I, J) = ( g^{-1} B_{1} g, g^{-1} B_{2} g, g^{-1} I, J g )
 \label{eq:ukac}
 \eeq
 on the operators $B_{1,2} \in End(K)$, $I: N \to K$, $J: K \to N$, where $K \approx {\BC}^{k}$ is a $k$-dimensional Hermitian vector space over $\BC$, and $N \approx {\BC}$ is a one-dimensional vector space. 
 \subsection{Family of integrable models}
 
 The space ${\CalM}_{k}$ is hyperk{\"a}hler, of real dimension
 $4k$ with a ${\CP}^{1}$-worth
 of complex structures and holomorphic symplectic forms. 
 Let $u \in {\BC\BP}^{1}$, represented by a pair $({\alpha}: {\beta})$ of complex numbers
 $u = ({\alpha}: {\beta})$, obeying $|{\alpha}|^{2} + |{\beta}|^{2}=1$, 
 and define two matrix-valued and two vector-valued sections of ${\CalO}(1)$ over ${\BC\BP}^{1}$:
 \beq
 \begin{aligned}
& B_{1}^{u} = {\alpha} B_{1} + {\beta} B_{2}^{\dagger} \, , \quad B_{2}^{u} = {\alpha} B_{2} - {\beta} B_{1}^{\dagger} \, , \\
& I^{u} = {\alpha} I + {\beta} J^{\dagger} \, , \quad J^{u} = {\alpha} J - {\beta} I^{\dagger} 
\end{aligned}
\label{eq:hkrot}
\eeq
obeying
\beq
[ B_{1}^{u}, B_{2}^{u} ] + I^{u} J^{u} = - {\alpha}{\beta} {\zeta} \cdot 
{\bf 1}_{K}
\label{eq:muc}
\eeq
and
\beq
[ B_{1}^{u}, (B_{1}^{u})^{\dagger}] + 
[B_{2}^{u}, (B_{2}^{u})^{\dagger} ] + I^{u} (I^{u})^{\dagger} - 
(J^{u})^{\dagger} J^{u} =  (|{\alpha}|^{2} - |{\beta}|^{2}) {\zeta} \cdot 
{\bf 1}_{K}
\label{eq:mur}
\eeq
The $u$-holomorphic functions on ${\CalM}_{k}$ are the 
$U(K)$-invariant holomorphic functions (polynomials) in 
$B_{1}^{u}, B_{2}^{u}, I^{u}, J^{u}$. It is well-known 
\cite{NikThesis, Gorsky:1999rb} that
${\CalM}_{k}$ is the phase space of the complexified
rational Calogero-Moser system, with or without the oscillator 
potential.  For the purposes of this note we would like to study the 
dynamics on ${\CalM}_{k}$ generated by 
\beq
H^{u} = {\rm tr}_{K} B_{1}^{u} B_{2}^{u}
\label{eq:hu2}
\eeq
with respect to the symplectic form ${\varpi}^{u}_{\BC}$ on ${\CalM}_{k}$
which descends from
\beq
{\Omega}^{u}_{\BC} = - \ii\, {\rm tr}\left( 
{\delta} B_{1}^{u} \wedge {\delta} B_{2}^{u} +
 {\delta} I^{u} \wedge {\delta} J^{u} \right)
\label{eq:omcu}
\eeq
by symplectic reduction. The latter can be described explicitly by 
going to the gauge where $X^{u}$, defined by
\beq
X^{u} = \frac{B_{1}^{u} + B_{2}^{u}}{\sqrt{2}} = {\alpha} {\bf X} -  {\ii\beta} {\bf P}^{\dagger} 
\label{eq:xmat}
\eeq
is diagonal
\beq
X^{u} = {\rm diag} \left( x_{1}, \ldots, x_{k} \right) \, , \ x_{i} \in {\BC}
\label{eq:xmat}
\eeq
so that $P$, defined by
\beq
P^{u}  = \frac{B_{1}^{u} - B_{2}^{u}}{\ii \sqrt{2}} = {\alpha} {\bf P} - {\ii\beta} {\bf X}^{\dagger}
\label{eq:pmat}
\eeq
is found to be the Lax operator of the rational Calogero-Moser system
\beq
P = \Vert p_{i} {\delta}_{ij} + \frac{\ii\nu}{x_{i} -x_{j}} (1-{\delta}_{ij}) \Vert_{i,j=1}^{k}\, , \qquad p_{i} \in {\BC}
\label{eq:plax}
\eeq
where ${\nu} = -  {\alpha\beta\zeta}$. The Hamiltonian \eqref{eq:hu2} becomes
\begin{multline}
H^{u} = \frac 12 \sum_{i=1}^{k} p_{i}^{2}  + {\CalU}({\bx}) \, , \\
{\CalU}({\bx}) =   \frac 12 \sum_{i=1}^{k} x_{i}^{2}  + \sum_{i < j} \frac{\nu^2}{(x_{i}-x_{j})^2}
\label{eq:hupx}
\end{multline}
while the symplectic form $\varpi^{u}$ is easily computed to be
\beq
{\varpi}^{u} = \sum_{i=1}^{k} dp_{i} \wedge dx_{i}
\eeq
The dynamics generated by \eqref{eq:hu2} looks simple before the reduction 
\beq
(B_1^u, B_2^u, I^u, J^u) \mapsto ( e^{\ii t} B_1^u , e^{-\ii t} B_2^u , I^u, J^u) 
\label{eq:cstar}
\eeq 

\subsection{Reality check} Now, for real $t$ the dynamics \eqref{eq:cstar} is actually an isometric action of the group $U(1)$, corresponding to the symmetry 
\beq
(B_{1}, B_{2}, I, J) \mapsto ( e^{{\ii} t} B_{1}, e^{-{\ii} t} B_{2}, I , J)
\label{eq:hbar}
\eeq 
of the original ADHM equations \eqref{eq:defadhm}, as is obvious from \eqref{eq:hkrot}. As shown in \cite{Nakajima:book} and used extensively in \cite{Nekrasov:2002qd}, the fixed points 
of the \eqref{eq:hbar} action are in one-to-one correspondence with partitions $\lambda$ of $k$, $|{\lambda}|=k$, with the linearization of \eqref{eq:hbar} on the tangent space to ${\CalM}_{k}$
given by the sum of $k$ copies of ${\BC}^{2}$ spaces, in one-to-one correspondence with the boxes ${\square} = (i,j)$, $1 \leq j \leq {\lambda}_{i}$, $1 \leq i \leq {\lambda}_{j}^{t}$ of $\lambda$,  where $U(1)$ acts with the weights $(h_{\square}, - h_{\square})$, where
\beq
h_{\square} = {\lambda}_{i} - j + {\lambda}_{j}^{t} - i +1
\label{eq:hook}
\eeq
is the hook-length of $\square$ \footnote{The fastest way to compute \eqref{eq:hook} uses a $U(1) \times U(1)$ isometric action on ${\CalM}_{k}$ of which only the $U(1)$ we discuss preserves the hyperk{\"a}hler structure}. It is interesting to note, that for $\zeta >0$, $J = 0$ (as demonstrated in \cite{Nakajima:book}), while $(B_1, B_2, I)$ can be described quite simply, modulo $GL(K)$-action, as operators of multiplication by $(z_1, z_2)$ and the image of function $1$, respectively, on the space $K = {\BC}[z_{1}, z_{2}]/{\CalI}$ associated with the codimension $k$ ideal $\CalI \subset {\BC}[z_1, z_2]$ in the ring of polynomials of two variables.  

The fixed points of the $U(1)$ action are also the stationary points of the Hamiltonian flow generated by \eqref{eq:hu2}.  The squares of the weights of the $U(1)$ action on the tangent space to the fixed point 
coincide with the eigenvalues of the Hessian $\frac 12 \frac{\partial^2 {\CalU}}{\partial x_{i} \partial x_{j}}$ of the potential  \eqref{eq:hupx} evaluated at its critical point, since the $p_i$-dependence of $H^u$ is a simple quadratic function\footnote{For more general Hamiltonian $H$ with a non-degenerate critical point $p$ the linearized dynamics near
$p$ has frequencies $\omega$, which are computed as the eigenvalues of ${\ii}{\varpi}^{-1} {\partial^{2} H} \vert_{p}$}.
It is a non-trivial and quite an interesting problem to describe explicitly the hermitian metric on $K$ for which \eqref{eq:defadhm} is satisfied, with $J =0$, $\zeta > 0$, and $(B_1, B_2, I)$ corresponding to the monomial ideal. In \cite{Braden:1999zp} this problem was solved, using the connection to the Calogero-Moser system, for special $\lambda = (1)^{k}$ or ${\lambda} = (k)$. 
It would be nice to try the insights of the present paper by G.~Felder and A.~Veselov to address the general case.

\subsection{Special $u$} Of all choices of the complex structure $u$ a circle $|{\alpha}| = | {\beta} | = \frac{1}{\sqrt{2}}$ is special, as the right hand side of the real moment map \eqref{eq:mur} vanishes. In this case finding a $GL(K)$ representative $(g^{-1} B_{1}^{u} g, g^{-1} B_{2}^{u} g, g^{-1} I^{u}, J^{u} g)$ solving \eqref{eq:mur} is equivalent to minimizing the norm
\begin{multline}
\Vert {\bf B}, {\bf I}, {\bf J} \Vert^2 = {\rm tr} \left( B_{1}^{u}(B_{1}^{u})^{\dagger} + B_{2}^{u} (B_{2}^{u})^{\dagger} + I^{u}(I^{u})^{\dagger} +  (J^{u})^{\dagger} J^{u} \right)
 = \\
 {\rm tr} \left( H^{-1} X H X^{\dagger} + H^{-1} P H P^{\dagger} + H^{-1} I I^{\dagger} \right)
 \label{eq:normh}
 \end{multline}
 where $X$ and $P$ are given by \eqref{eq:xmat}, \eqref{eq:plax}, with
 \beq
 X^{\dagger} = {\rm diag}( {\bar x}_{1}, \ldots, {\bar x}_{k} ) \, , \ P^{\dagger} = \Vert {\bar p}_{i} {\delta}_{ij} +  \frac{\ii {\bar \nu}}{{\bar x}_{i} - {\bar x}_{j}} (1-{\delta}_{ij}) \Vert_{i,j=1}^{k} 
 \label{eq:xpmatherm}
 \eeq
 and $H$ is the Hermitian matrix in $K$ to be found from minimizing \eqref{eq:normh}.

\subsection{Generalizations}

There are several generalizations of the Calogero-Moser system associated with quiver varieties and (generalized) moduli spaces of instantons. One generalization, the so-called spin CM system, is defined on the moduli space ${\CalM}_{k}(n)$ of rank $n$ framed torsion free sheaves on ${\BC\BP}^{2}$, or the $U(n)$ instantons on the noncommutative ${\BR}^{4}_{\theta}$. The difference with \eqref{eq:defadhm} is that the space $N$ is now $n$-dimensional, and $I: N \to K$ and $J: K \to N$ are non-zero even for $\zeta >0$. Instead of a many-body system of $k$ particles one gets a system of $k$ particles with ``spin'', with the phase space being a resolution of singularities of $\left( S_{n-1} \times {\BC}^{2} \right)^{k}/S(k)$, with $S_{n-1}$ being the  space of equivalence classes of pairs $(v, u)$, where $u \in N$ is a vector, $v \in N^{*}$ a covector, $v(u) = {\nu}$ and $(v,u) \sim (t v, t^{-1} u)$, $t \in {\BC}^{\times}$.  The Hamiltonian, generating the action of the maximal torus ${\BC}^{\times} \times ( {\BC}^{\times} )^{n-1}$, is given by
\beq
H^{u} = {\rm tr} B_{1}^{u} B_{2}^{u} + {\rm tr} I A J
\eeq
where $A = {\rm diag}(a_{1}, \ldots, a_{n})$ is a diagonal matrix acting in the $N$ space. Upon solving the moment map equations, again, in the gauge where $X$ is diagonal, we find the Hamiltonian describing $k$ spins interacting both with each other and the external magnetic field: 
\beq
H^{u} = \frac 12 \sum_{i=1}^{k} \left( p_{i}^{2} + x_{i}^{2}  \right) + \sum_{i < j}  \sum_{{\alpha},{\beta}=1}^{n} \frac{S_{i \alpha}^{\beta} S_{j \beta}^{\alpha}}{(x_{i}-x_{j})^2} + \sum_{\alpha = 1}^{n} a_{\alpha} \sum_{i=1}^{k} S_{i \alpha}^{\alpha}
\label{eq:spin}
\eeq
where ${\bf S}_{i} = \left( S_{i \alpha}^{\beta} \right)$ for $i = 1, \ldots, k$ are the generators of $k$ copies of $\mathfrak{sl}_{n}$ obeying
\beq
\{ S_{i \alpha}^{\beta} , S_{j \alpha'}^{\beta'} \} = {\delta}_{ij} \left( {\delta}^{\beta}_{\alpha'} S_{i \alpha}^{\beta'}  - {\delta}^{\beta'}_{\alpha} S_{i \alpha'}^{\beta} \right)
\eeq
with ${\bf S}_{i} - \frac{\ii \nu}{n} {\bf 1}_{N} = \Vert I_{i \alpha} J^{\beta}_{i} \Vert_{\alpha, \beta = 1}^{n}$ being rank one matrix.  We deduce using \cite{Nekrasov:2002qd} that the stationary configurations of \eqref{eq:spin} are in one-to-one
correspondence with $n$-tuples $({\lambda}^{({\alpha})})_{\alpha = 1}^{n}$ of partitions, obeying
\beq
\sum_{\alpha = 1}^{n} | {\lambda}^{({\alpha})} | = k
\eeq
The $2 n k$ eigen-frequencies of oscillations near the stationary points are given by (again, borrowing from \cite{Nekrasov:2002qd} the calculation of the weights of the torus action on the tangent space to ${\CalM}_{k}(n)$ at the fixed point $({\lambda}^{({\alpha})})$):
\beq
\pm ( a_{\alpha} - a_{\beta} + i + j - 1 - {\lambda}^{({\beta})t}_{j} - {\lambda}^{({\alpha})}_{i} )\, , 1 \leq {\alpha}, {\beta} \leq n\, , \ (i,j) \in {\lambda}^{({\beta})}  
\eeq 
Another generalization is the higher $g > 1$ genus analogue ${\mathfrak{M}}_{g,k}$ of the quiver variety
(it arises in the studies of the Coulomb branches of the moduli spaces of vacua of three dimensional theories with ${\CalN}=4$ supersymmetry). In its simplest form, it is the space of stable $2g$-tuples of $k \times k$ matrices $\mathfrak{A}_{I}, \mathfrak{B}_{I}$, $I = 1, \ldots , g$, obeying
\beq
\sum_{I=1}^{g} [ \mathfrak{A}_{I}, \mathfrak{B}_{I} ] = 0
\eeq
modulo the $GL(k)$-symmetry $(\mathfrak{A}_{I}, \mathfrak{B}_{I})_{I=1}^{g} \mapsto ( G^{-1} \mathfrak{A}_{I} G,  G^{-1}\mathfrak{B}_{I} G)_{I=1}^{g}$. It is a holomorphic symplectic manifold with the symplectic form descending from
\beq
{\Omega} = \sum_{I=1}^{g} {\rm tr} {\delta} {\mathfrak{A}}_{I}  \wedge  {\delta} {\mathfrak{B}}_{I} 
\eeq
We define the Hamiltonian 
\beq
H = \sum_{I=1}^{g} {\ve}_{I} {\rm tr} \left(  \mathfrak{A}_{I} \mathfrak{B}_{I} \right)
\eeq
which generates a one-parametric subgroup in the torus $({\BC}^{\times})^{g}$ of symplectic symmetries 
of ${\mathfrak{M}}_{g,k}$. We have no idea how to effectively classify its stationary points, though. Perhaps the ideas from \cite{Ko2} might be useful.


\begin{thebibliography}{}

\bibitem{Adler}
V.E. Adler {\it A modification of Crum's method.} Theor. Math. Phys. {\bf 101 (3)} (1994), 1381-87.

\bibitem{AM} 
M. Adler, J. Moser {\it On a class of polynomials connected with the Korteweg-de Vries equation.} Commun. Math. Phys., {\bf 61}, 1-30, (1978).

\bibitem{ABCOP}
S. Ahmed, M. Bruschi, F. Calogero, M.A. Olshanetsky and A.M. Perelomov {\it Properties of the zeros of the classical polynomials and of the Bessel functions.} 
Il Nuovo Cimento {\bf 49(2)} (1979), 173-199.

\bibitem{AMM}
H. Airault, H.P. McKean and J. Moser {\it Rational and elliptic solutions of
the Korteweg-de Vries equation and a related many-body problem.}
Comm. Pure Appl. Math. {\bf 30} (1977), 95--178.

\bibitem{Appell}
P. Appell {\it Sur une classe de polynomes.}  Annales Sci. \'Ecole Norm. Supér. 2e S\`erie. {\bf 9} (1880), 119-144.

\bibitem{BLZ}
V.V. Bazhanov, S.L. Lukyanov and A.B. Zamolodchikov {\it Higher level eigenvalues of $Q$ operators and Schr\"odinger equation.} Adv. Theor. Math. Phys. {\bf 7} (2003), 711.

\bibitem{BHSS}
N. Bonneux, Z. Hamker, J. Stembridge, M. Stevens {\it Wronskian Appell polynomials and symmetric
functions.} Advances in Applied Math. {\bf 111} (2019), 101932.

\bibitem{BCh}
J.L. Burchnall, T.W. Chaundy {\it A set of differential equations which 
can be solved by polynomials}, Proc. London Math. Soc. 
{\bf 30} (1929-30), 401--414.

\bibitem{Cal71}
F. Calogero {\it Solution of the one-dimensional n-body problems with quadratic and/or
inversely quadratic pair potentials.} J. Math. Phys. {\bf 12} (1971), 419-436.

\bibitem{Cal78}
F. Calogero {\it Equilibrium configuration of the one-dimensional $n$-body problem with quadratic and inversely quadratic pair potentials.} Lett. Nuovo Cimento, {\bf 20} (1978), 251.

\bibitem{C}
F. Calogero {\it Disproof of the conjecture.} Lett. Al Nuovo Cimento, {\bf 35(6)} (1982), 181-185.

\bibitem{CM}
R. Conti, D. Masoero {\it Counting monster potentials.} J High Energy Physics, JHEP02 (2021), 059. ArXiv ePrint: 2009.14638.


\bibitem{DG}
J.J.Duistermaat, F.A.Gr\"unbaum {\it Differential equations in the
spectral parameter.} Comm. Math. Phys. {\bf 103} (1986), 177-240.

\bibitem{Duran} A.J. Duran {\it A proof of the Veselov Conjecture for segments.} Proc. Amer. Math. Soc. {\bf 149} (2021), 173-188.

\bibitem{Frob}
G. Frobenius {\it \"Uber die Charaktere der symmetrischen Gruppe.} Sitzungsberichte der
Akademie der Wiss. zu Berlin, pages 516–534, 1900.

\bibitem{Grosu}
C. Grosu, C. Grosu {\it The irreducibility of some Wronskian Hermite polynomials.} Indagationes Mathematicae {\bf 32:2} (2021), 456--497.



\bibitem{FHV12} G. Felder, A.D. Hemery, A.P. Veselov {\it Zeros of Wronskians of Hermite polynomials and Young diagrams}, Physica D {\bf 241} (2012), 2131-2137.

\bibitem{KKS}
D. Kazhdan, B. Kostant, S. Sternberg {\it Hamiltonian group actions and dynamical systems
of Calogero type.} Comm. Pure Appl. Math. {\bf 31} (1978), 481-507.

\bibitem{Mac}
I.G. Macdonald {\it Symmetric functions and Hall polynomials.} Oxford Univ. Press, 1995.

\bibitem{M}
J.~Moser {\it Three integrable Hamiltonian systems connected with isospectral deformations}.
Advances in Math. 16 (1975), 197--220.

\bibitem{NN}
N. Nekrasov {\it Private communication}. July 2, 2024, BIMSA.

\bibitem{Oblomkov} 
A.A. Oblomkov {\it Monodromy-free Schr\"odinger operators with quadratically increasing potentials.} Theor. Math. Phys., {\bf 121} (1999), 1574-1584.

\bibitem{OP}
M.A. Olshanetsky, A.M. Perelomov {\it Classical integrable finite-dimensional systems related to Lie algebras.}
Physics Reports {\bf 71:5} (1981), 313-400.

\bibitem{Per}
A.M. Perelomov, Preprint ITEP, N27, 1976.

\bibitem{Per2}
A.M. Perelomov {\it Equilibrium configurations and small oscillations of some dynamical systems.}
Annales de l’I. H. P., section A, tome 28, no 4 (1978), p. 407-415.

\bibitem{Sz}
G. Szeg\H o {\it Orthogonal Polynomials.} American Mathematical Society, New York (1939).

  \bibitem{V}
A.P. Veselov {\it On Stieltjes relations, Painlev\`e-IV hierarchy and complex monodromy.} J. Phys. A 34 (2001), no. 16, 3511--3519.  

\bibitem{Wilson}
G. Wilson {\it Collisions of Calogero--Moser particles and an adelic Grassmannian (With an Appendix by I.G. Macdonald).}
Invent. Math. {\bf 133} (1998), 1-41.


\end{thebibliography}

\begin{thebibliography}{99}
 
 \bibitem{Ko2}
 M.~Kontsevich, \emph{Formal  (non)-commutative symplectic geometry}, The Gelfand Mathematical Seminars, 1990- 1992, Birkh{\"a}user (1993), pp. 173- 187
 

\bibitem{Nakajima:1995}
H.~Nakajima, \emph{Heisenberg algebra and Hilbert scheme of points on projective
surfaces}, alg-geom/9507012v2
 
\bibitem{NikThesis}

N.~Nekrasov, \emph{Four dimensional holomorphic theories}, PhD. thesis, Princeton University, 1996, available from 
{\tt https://media.scgp.stonybrook.edu/papers/prdiss96.pdf}
 
\bibitem{Nekrasov:1998ss}
N.~Nekrasov and A.~S.~Schwarz,
\emph{Instantons on noncommutative ${\BR}^{4}$ and $(2,0)$-superconformal six-dimensional theory},
Commun. Math. Phys. \textbf{198}, 689-703 (1998)
doi:10.1007/s002200050490
[arXiv:hep-th/9802068 [hep-th]].


\bibitem{Gorsky:1999rb}
A.~Gorsky, N.~Nekrasov and V.~Rubtsov,
\emph{Hilbert schemes, separated variables, and D-branes},
Commun. Math. Phys. \textbf{222}, 299-318 (2001)
doi:10.1007/s002200100503
[arXiv:hep-th/9901089 [hep-th]].

\bibitem{Braden:1999zp}
H.~W.~Braden and N.~A.~Nekrasov,
\emph{Space-time foam from noncommutative instantons},
Commun. Math. Phys. \textbf{249}, 431-448 (2004)
doi:10.1007/s00220-004-1127-2
[arXiv:hep-th/9912019 [hep-th]].



\bibitem{Nakajima:book}
H.~Nakajima, \emph{Lectures on Hilbert schemes of points on surfaces}, 
AMS University Lecture Series, 1999, vol. {\bf 18}, 
DOI: {\tt https://doi.org/10.1090/ulect/018}

\bibitem{Nekrasov:2002qd}
N.~A.~Nekrasov,
\emph{Seiberg-Witten prepotential from instanton counting},
Adv. Theor. Math. Phys. \textbf{7}, no.5, 831-864 (2003)
doi:10.4310/ATMP.2003.v7.n5.a4
[arXiv:hep-th/0206161 [hep-th]].


\end{thebibliography}
\end{document}